\documentclass{amsart}[12pt]
\usepackage{enumerate, amsthm, mathrsfs, hyperref}

\newtheorem{theorem}{Theorem}
\newtheorem{lemma}{Lemma}
\newtheorem{proposition}{Proposition}

\newcommand{\im}{\operatorname{im}}

\title{A note on the ultrahyperbolic wave equation in $3+3$ dimensions}

\author{Jonathan Holland}
\begin{document}
\maketitle

This note concerns the following two operators on $S^3\times S^3$ recently studied in connection with twistor theory in \cite{Sparling1}, \cite{Sparling2}.  For $x,y\in S^3$ and $f\in C^\infty(S^3\times S^3)$, define
\begin{align}
\label{T}Tf(x,y) &= \int_{S^3} f(xg,gy)\,dg\\
\label{Box}\Box f &= \Delta_x f - \Delta_y f
\end{align}
where $dg$ is the invariant probability measure on the Lie group $S^3\cong SU(2)$ and $\Delta$ is the Laplace-Beltrami operator on $S^3$.  Our main theorem is the following

\begin{theorem}
Acting on smooth functions, $\ker T = \im \Box$ and $\ker\Box = \im T$;  that is, $T, \Box$ form an exact couple.
\end{theorem}

The strategy is to first prove that the theorem holds for $T$ and $\Box$ restricted to the subspace $\mathscr{H}$ consisting of finite linear combinations of spherical harmonics, and then to extend this first, by a density argument, to the Sobolev spaces $H^s = H^s(S^3\times S^3)$.  The result then follows from the Sobolev lemma $C^\infty=\bigcap_s H^s$.  Both operators are understood as acting on $H^s$ in a weak sense.

We first review some of the theory of spherical harmonics in order to fix notation.  Let $\mathscr{P}_k$ be the space of all polynomials $f$ on $\mathbb{R}^4$ homogeneous of degree $k$: $f(tx) = t^k f(x)$ for all $t\in \mathbb{R}$ and $x\in \mathbb{R}^4$.  Let $\mathscr{H}_k\subset \mathscr{P}_k$ be the subset consisting of all harmonic polynomials: $\Delta_{\mathbb{R}^4} f = 0$, where $\Delta_{\mathbb{R}^4}$ is the ordinary Euclidean Laplacian.  The following theorem summarizes the basic properties of spherical harmonics:
\begin{proposition}\mbox{}
\begin{enumerate}[(a)]
\item Every polynomial $p\in \mathscr{P}_k$ can be written uniquely as
$$p(x) = p_0(x) + |x|^2 p_1(x) + \cdots + |x|^{\lfloor k/2\rfloor} p_{\lfloor k/2\rfloor}$$
where $p_j\in \mathscr{H}_{k-2j}$.  
\item If $p\in\mathscr{H}_k$, then $\Delta_{S^3} p|_{S^3} = -k(k+2) p|_{S^3}$, where $\Delta_{S^3}$ is the Laplace-Beltrami operator on the unit sphere in $\mathbb{R}^4$.
\end{enumerate}
\end{proposition}

By the Stone-Weierstrass theorem, spherical harmonics are dense in $C(S^3)$, and hence also in $L^2(S^3)=H^0(S^3)$.  The different eigenspaces of the Laplace-Beltrami operator are orthonormal in $L^2(S^3)$, and so the spaces $\mathscr{H}_k$ give an orthogonal decomposition
$$L^2(S^3) = \widehat{\bigoplus_{k=0}^\infty} \mathscr{H}_k$$
in terms of which it is possible to decompose uniquely any $f\in L^2(S^3)$ into the Fourier series
\begin{equation}\label{Fourier}
f = \sum_{k=0}^\infty a_k Y_k
\end{equation}
where $Y_k\in \mathscr{H}_k$ are normalized to have unit norm.  In particular, the spaces $\mathscr{H}_k$ exhaust the eigenspaces of $\Delta_{S^3}$. The Sobolev spaces admit the following characterization in terms of the Fourier series:

\begin{lemma}
Let $s\ge 0$ and $f\in H^0(S^3)$ have Fourier decomposition \eqref{Fourier}.  Then $f\in H^s(S^3)$ if and only if $\sum_{k=0}^\infty |a_k|^2 k^{2s} < \infty$.
\end{lemma}

Now we turn attention to $S^3\times S^3$.  The eigenvalues of $\Delta_{S^3\times S^3}=\Delta_x + \Delta_y$ are precisely the tensor products $\mathscr{H}_k\otimes\mathscr{H}_\ell$.  Thus to each $f\in H^s(S^3\times S^3)$, there is a double series expansion in spherical harmonics
$$f = \sum_{k,\ell = 0}^\infty a_{k\ell} Y_k\otimes Y_\ell$$
where the $Y_k$ are normalized elements of $\mathscr{H}_k$.  As a result, there is an isomorphism of Hilbert spaces
$$H^s(S^3\times S^3) \cong H^s(S^3)\,\widehat{\otimes}\, H^s(S^3),$$
where the tensor product appearing on the right is the completion of the usual tensor product under the Hilbert inner product defined on simple tensors by $\langle a\otimes b, c\otimes d\rangle_s = \langle a, c\rangle_s \langle b, d\rangle_s$.

\begin{lemma}\label{kernel} \mbox{}
\begin{enumerate}[(a)]
\item The operator $T:H^s(S^3\times S^3)\to H^s(S^3\times S^3)$ is a bounded self-adjoint operator.  
\item The kernel of $T$ contains the subspace $\widehat{\bigoplus}_{k\not=\ell} \mathscr{H}_k\otimes\mathscr{H}_\ell$, where the closure is taken in the $H^s$ norm.
\end{enumerate}
\end{lemma}

\begin{proof}
\begin{enumerate}[(a)]
\item It is sufficient to prove boundedness and self-adjointness in $L^2$, since the proof for $H^s$ follows by applying the $L^2$ result to $(1-\Delta_x-\Delta_y)^{s/2}f$.  By the integral Minkowski inequality,
\begin{align*}
\|Tf\|_2 &= \left\{\int_{S^3\times S^3}\left|\int_{S^3} f(xg,gy)\,dg\right|^2\,dxdy\right\}^{1/2} \\
&\le \int_{S^3}\left\{\int_{S^3\times S^3} |f(xg,gy)|^2\,dxdy\right\}^{1/2}\,dg = \|f\|_2.
\end{align*}
This proves boundedness.  Self-adjointness follows from an application of Fubini's theorem.
\item Suppose first that $f(x,y) = u(x)v(y)$ where $u\in\mathscr{H}_k$, $v\in\mathscr{H}_\ell$, and $k\not=\ell$.  Since $\Delta_{S^3}$ is invariant under both left and right translation in $S^3$, for fixed $x,y\in S^3$ we have $R_xu\in \mathscr{H}_k$ and $L_yv\in\mathscr{H}_\ell$ so that by orthogonality of the $\mathscr{H}_k$,
$$Tf(x,y) = \int_{S^3} R_xu L_yv = 0.$$
The result now follows from the boundedness of $T$.
\end{enumerate}
\end{proof}

A sharper result, completely classifying the range of $T$ as well, requires knowing that the range is closed.  Let $E_{k,\ell}:H^s\to \mathscr{H}_k\otimes\mathscr{H}_{\ell}$ be the orthogonal projection.  The following gives the operator $T$ as a Fourier multiplier:

\begin{theorem}
For any $s$, $T\in \mathscr{B}(H^s,H^{s+1})$.  The kernel of $T$ is $\widehat{\bigoplus}_{k\not=\ell}\mathscr{H}_k\otimes\mathscr{H}_\ell$ and the image is $\widehat{\bigoplus}_k \mathscr{H}_k\otimes \mathscr{H}_k$.  In fact,
$$T = \sum_{k=0}^\infty \frac{1}{k+1}R_kE_{k,k}$$
where each $R_k$ is a reflection operator (self-adjoint involution) in the space $\mathscr{H}_k\otimes\mathscr{H}_k$.
\end{theorem}

Assume for the moment that the image of $T\in \mathscr{B}(H^{s-1},H^s)$ contains the space $B_s=\widehat{\bigoplus}_k \mathscr{H}_k\otimes \mathscr{H}_k$.  By Lemma \ref{kernel}, the kernel of $T$ in $H^s$ contains the subspace $A_s=\widehat{\bigoplus}_{k\not=\ell} \mathscr{H}_k\otimes \mathscr{H}_\ell$.  To prove the opposite inclusion, note that the spaces $A_s$ and $B_s$ are orthogonal complements.  So any $x\in\ker T$ has the form $x=a\oplus b$ for $b\in B_s$.  By assumption, $b=Tc$ for some $c\in H^{s-1}$.  Then applying $T$ gives $0=Tx=T^2c$, and this implies that $c=0$ because $T$ is self-adjoint.  Therefore $x\in A_s$, and so $\ker T=A_s$, as claimed.  Effectively the same argument shows that, under the same assumption, the image of $T$ on $H^{s-1}$ must then be {\em equal} to $B_s$.  It is therefore enough to show that the image of $T$ in $H^s$ contains the subspace $B_s$.

The proof of this fact along with last assertion of the theorem employs the zonal spherical harmonics $Z^{(k)}_x(y)$ on $S^3$, defined for $f\in C^\infty(S^3)$ by
$$\int_{S^3} Z_x^{(k)}(y) f(y)\,dy = (E_kf)(y)$$
where $E_k$ is the orthogonal projection onto $\mathscr{H}_k$.  They satisfy the following properties

\begin{lemma}[\cite{SteinWeiss}]\label{Zonals}\mbox{}
\begin{enumerate}[(a)]
\item For each fixed $x$, $y\mapsto Z_x^{(k)}(y)$ is in $\mathscr{H}_k$.
\item For each rotation $\rho\in SO(4)$, $Z_{\rho x}^{(k)}(\rho y) = Z_{x}^{(k)}(y)$.
\item Conversely, any function satisfying these two properties is a constant multiple of $Z^{(k)}_x(y)$.
\item The following integral identities hold:
$$(k+1)^2 = \int_{S^3} |Z^{(k)}_y(x)|^2\,dx = \int_{S^3} Z^{(k)}_x(x)\,dx = Z^{(k)}_e(e).$$
\end{enumerate}
\end{lemma}

The zonal harmonics enter the proof by integrating $Tf$, for $f\in \mathscr{H}_k\otimes\mathscr{H}_k$ against the product of zonal harmonics $Z^{(k)}_{x'}(x)Z^{(k)}_{y'}(y)$.  Fubini's theorem, followed by an obvious change of variables gives
\begin{align*}
\iint_{S^3\times S^3}Z^{(k)}_{x'}(x)Z^{(k)}_{y'}(y)Tf(x,y)\,dxdy&= \int_{S^3} \iint_{S^3\times S^3} Z^{(k)}_{x'}(x)Z^{(k)}_{y'}(y)f(xg,gy)\,dxdydg\\
&= \int_{S^3} \iint_{S^3\times S^3} Z^{(k)}_{x'}(xg^{-1})Z^{(k)}_{y'}(g^{-1}y)f(x,y)\,dxdydg\\
&=\iint_{S^3\times S^3}  f(x,y)\,dxdy\int_{S^3} Z^{(k)}_{x'g}(x)Z^{(k)}_{gy'}(y)dg \\
&\qquad\text{by invariance of $Z^{(k)}$}\\
&=\iint_{S^3\times S^3}  f(x,y)\,dxdy\int_{S^3} Z^{(k)}_{g}(\bar{x}'x)Z^{(k)}_{gy'}(gy)dg\\
&=\iint_{S^3\times S^3}  f(x,y)\,dxdy\int_{S^3} Z^{(k)}_{\bar{x}'xy'}(y)dg
\end{align*}
by the reproducing property of $Z^{(k)}$.  The bar denotes quaternionic conjugation, which is the same as inversion in the group $S^3$.  Thus for $f\in \mathscr{H}_k\otimes\mathscr{H}_k$, $Tf$ is obtained by integrating against the kernel
$$Z^{(k)}_{\bar{x}'xy'}(y) = Z^{(k)}_{\bar{x}'x}(y\bar{y}') = Z^{(k)}_{xy'}(x'y).$$

To establish the theorem, it is sufficient to show that $T^2|_{\mathscr{H}_k\otimes\mathscr{H}_k} = \lambda_k^2\operatorname{Id}_{\mathscr{H}_k\otimes\mathscr{H}_k}$.  Applying $T$ twice is the same as integrating against the kernel in the variables $(x'',y'')$ defined by
$$K(x'',y''; x,y) = \iint_{S^3\times S^3} Z^{(k)}_{x'y''}(x''y') Z^{(k)}_{xy'}(x'y)\,dx'dy'.$$
To show that $T^2|_{\mathscr{H}_k\otimes\mathscr{H}_k} = \lambda_k^2 \operatorname{Id}$, it is sufficient to prove that $K(x'',y'';x,y) = \lambda_k^2 Z_x(x'')Z_y(y'')$.  By the characterization in Lemma \ref{Zonals}(c) it is then enough to show that for all $g\in S^3$,
\begin{align*}
K(x'',y''; x,y) &= K(x''g,y'';xg,y) =  K(gx'',y'';gx,y)\\
&= K(x'',y''g;x,yg) = K(x'',gy'';g,gy),
\end{align*}
since every special orthogonal transformation has the form $x\mapsto gxh$ on $S^3$ for some $g,h$.  Clearly, by symmetry in the variables, it is enough to show the first two identities.  By invariance of the Haar measure,
\begin{align*}
 K(x''g,y'';xg,y) &=  \iint_{S^3\times S^3} Z^{(k)}_{x'y''}(x''gy') Z^{(k)}_{xgy'}(x'y)\,dx'dy'\\
&= \iint_{S^3\times S^3} Z^{(k)}_{x'y''}(x''\tilde{y}') Z^{(k)}_{x\tilde{y}'}(x'y)\,dx'd\tilde{y}' = K(x'',y''; x,y)
\end{align*}
and,
\begin{align*}
 K(gx'',y'';gx,y) &=  \iint_{S^3\times S^3} Z^{(k)}_{x'y''}(gx''y') Z^{(k)}_{gxy'}(x'y)\,dx'dy'\\
&=  \iint_{S^3\times S^3} Z^{(k)}_{g^{-1}x'y''}(x''y') Z^{(k)}_{xy'}(g^{-1}x'y)\,dx'dy'\\
&= \iint_{S^3\times S^3} Z^{(k)}_{\tilde{x}'y''}(x''y') Z^{(k)}_{xy'}(\tilde{x}'y)\,d\tilde{x}'dy' = K(x'',y''; x,y).
\end{align*}

So on $\mathscr{H}_k\otimes\mathscr{H}_k$, $T^2=\lambda_k^2\operatorname{Id}$.

It remains to calculate $\lambda_k$.  For this, Lemma \ref{Zonals}(d) gives
\begin{align*}
(k+1)^4\lambda_k^2&=\iint_{S^3\times S^3} K(x,y;x,y)\,dxdy \\
&= \iiiint_{(S^3)^4}Z^{(k)}_{x'y}(xy') Z^{(k)}_{xy'}(x'y)\,dx'dy'dxdy \\
&= \iiiint_{(S^3)^4}Z^{(k)}_{x'}(xy'\bar{y}) Z^{(k)}_{xy'}(x'y)\,dx'dy'dxdy \\
&= \iiint_{(S^3)^3}Z^{(k)}_{xy'}(xy')\,dy'dxdy \\
&= (k+1)^2
\end{align*}
so that $\lambda_k=(k+1)^{-1}$, as required.

\bibliography{xi_transform}{}
\bibliographystyle{kp}

\end{document}